\documentclass{article}
\usepackage{amssymb}
\usepackage{amsmath}
\usepackage{amsfonts}
\usepackage{bbm}
\usepackage{mathrsfs}
\usepackage{tikz}
\usepackage{pgflibraryarrows}
\usepackage{pgflibraryplotmarks}
\usepackage{hyperref}

\usepackage[letterpaper]{geometry}

\newcommand{\R}{\mathbb{R}}
\newcommand{\E}{\mathbb{E}}
\newcommand{\N}{\mathbb{N}}

\newcommand{\Z}{\mathbb{Z}}

\renewcommand{\P}{\mathbb{P}}
\newcommand{\s}{\sigma}

\newtheorem{theorem}{Theorem}[section]
\newtheorem{remark}[theorem]{Remark}
\newtheorem{proposition}[theorem]{Proposition}
\newtheorem{lemma}[theorem]{Lemma}

\newtheorem{corollary}[theorem]{Corollary}

\newenvironment{proof}{{\bf Proof:}}{\hfill$\square$\vskip.5cm}
\newenvironment{proofof}{}{\hfill$\square$\vskip.5cm}

\def\be{\begin{equation}}
\def\ee{\end{equation}}
\def\bea{\begin{eqnarray}}
\def\eea{\end{eqnarray}}
\def\<{\langle}
\def\>{\rangle}

\title{Factorization properties in $d$-dimensional spin glasses.\\ Rigorous results and some perspectives.}

\author{Pierluigi Contucci\footnote{Universit\`a di Bologna, Piazza di Porta S.Donato 5, 40127 Bologna, Italy} \and  Emanuele Mingione \footnote{Universit\`a di Bologna, Piazza di Porta S.Donato 5, 40127 Bologna, Italy}  \and Shannon Starr\footnote{University of Alabama at Birmingham, 1300 University Boulevard, Birmingham, AL 35294-1170, USA}}

\date{\today}

\begin{document}

\maketitle

\begin{abstract}
\noindent
In this paper we show that $d$-dimensional Gaussian spin glass models
are strongly stochastically stable, fulfill the Ghirlanda-Guerra identities in distribution
and the ultrametricity property.
\end{abstract}

\vspace{8pt}
\vskip .2 cm
\noindent
{\small \bf MCS numbers:} Primary 60B10, 60G57, 82B20; Secondary 60K35.

\maketitle

\section{Introduction}
After the important result by Panchenko \cite{Panchenkoultra} on the proof of ultrametricity for a class
of mean field spin glasses the next natural question is to understand if its validity can also
be established for a $d$-dimensional lattice model more closely related to physical
systems. Among those the Edwards Anderson model defined by the Hamiltonian on $\Lambda\subset\Z^d$
\begin{equation}\label{hea}
H_{\Lambda} (\sigma) \; = \;
- \sum_{\substack{\boldsymbol{i},\boldsymbol{j} \in \Lambda\subset\Z^d\\ |\boldsymbol{i}-\boldsymbol{j}|=1}}
J_{\boldsymbol{i}\boldsymbol{j}}\sigma_{\boldsymbol{i}} \sigma_{\boldsymbol{j}}
\end{equation}
has certainly played a major role \cite{EA}. That model has escaped so far
a rigorous mathematical analysis. From the physical point of view the open question is
{\it do and (if yes) to what extent finite dimensional spin glasses behave according to
the mean field theory} \cite{MPV}. The question can be made more explicit in terms of
the overlap (Hamiltonian covariance) probability distribution:
\begin{enumerate}
\item the overlap distribution has a non trivial support.
\item the multi-overlap distribution factorizes into the single one according to the Parisi's
ultrametric scheme.
\end{enumerate}

The present paper deals with question $2)$ and answers positively to it. 

Let us first illustrate the statistical physics ideas that we are following.
The factorization laws that we deal with in this paper can be understood
as consequences of a simple stability method. Stability in Statistical Mechanics works by identifying
a small (yet non-trivial) deformation of the system, prove that in the
large volume limit the perturbation vanishes and, by means of the linear response theory,
compute the relations among observable quantities. This method leads to interesting
consequences and applications because it reduces the {\em a priori} degrees of freedom of a theory.
Following the ideas developed in \cite{cggep} one starts, in classical models, for smooth
bounded functions $f$ of spin configurations, from the counting measure
\begin{equation}
\mu_N(f)  =  \frac{1}{2^N}\sum_{\sigma}f(\sigma)  ,
\end{equation}
and defines the equilibrium state
\begin{equation}
\omega_{\beta,N}(f)  =  \frac{\mu_N(fe^{-\beta H_N})}{\mu_N(e^{-\beta H_N})}  .
\end{equation}
By considering the Hamiltonian per particle
\begin{equation}
h_N(\sigma)=\frac{H_N(\sigma)}{N}
\end{equation}
the classical perturbed state is defined by
\begin{equation}
\omega_{\beta,N}^{(\lambda)}(f)  =  \frac{\omega_{\beta,N}(fe^{-\lambda h_N})}{\omega_{\beta,N}(e^{-\lambda h_N})}  .
\end{equation}
Since the perturbation amounts to a small change in the temperature
\begin{equation}
\omega_{\beta,N}^{(\lambda)}(f)  =  \omega_{\beta +\frac{\lambda}{N},N}(f)
\end{equation}
one has that, apart from isolated singularity points, in the thermodynamic limit
\begin{equation}
\frac{d \omega_{\beta,N}^{(\lambda)}(f)}{d\lambda}  \rightarrow  0  .
\end{equation}
One may appreciate the content of the previous property by showing that it implies, for the Curie-Weiss ferromagnetic model in zero magnetic field, the relation
\begin{equation}
\omega_{\beta}(\sigma_1\sigma_2\sigma_3\sigma_4)  =  \omega_{\beta}(\sigma_1\sigma_2)^2  \; .
\end{equation}
Hence, although the magnetization itself may fail to concentrate due to a spin-flip symmetry
breaking, the square of the magnetization does concentrate in the thermodynamic limit.

The previous approach leads to the concept of {\it Stochastic Stability} when applied,
suitably adapted, to the spin glass phase. Consider, for smooth bounded function $f$ of $n$ spin
configurations, the quenched equilibrium state
\begin{equation}
< f >_{\beta,N}  =
\mathbb{E}\left(\frac{\sum_\sigma f(\sigma)e^{-\beta H_N}}{\sum_\sigma e^{-\beta H_N}}\right)\, .
\end{equation}
Define the deformation as:
\begin{equation}
< f >^{(\lambda)}_{\beta,N}  =  \frac{<fe^{-\lambda h_N}>}{<e^{-\lambda h_N}>}  .
\end{equation}
We observe that the previous deformation is, unlike in the classical case, not a simple temperature
shift. In fact:
\begin{equation}
< f >^{(\lambda)}_{\beta,N}  =
\frac{\mathbb{E}\left(\frac{\sum_\sigma f(\sigma)e^{-(\beta+\lambda/N)H_N}}{\sum_\sigma e^{-\beta H_N}}\right)}
{\mathbb{E}\left(\frac{\sum_\sigma e^{-(\beta+\lambda/N)H_N}}{\sum_\sigma e^{-\beta H_N}}\right)}\, ;
\end{equation}
nevertheless, the system is still stable with respect to it in a sense that will be made precise
in the following sections and is essentially captured by saying that apart from isolated singularity points,
in the thermodynamic limit
\begin{equation}
\frac{d}{d\lambda} < f >^{(\lambda)}_{\beta,N}  \to  0\, .
\end{equation}
Moreover the previous stability property implies (by use of the integration by parts techinque)
that the following set of identities (Ghirlanda-Guerra), first derived in \cite{GhirlandaGuerra}, holds:
\begin{equation}
< f c_{1,n+1}>_{\beta,N}   =  \frac{1}{n}< f >_{\beta,N}< c >_{\beta,N} +
\end{equation}
\begin{equation}\nonumber
+ \frac{1}{n}\sum_{j=2}^{n}< fc_{1,j} >_{\beta,N}  \; ,
\end{equation}
where the term $c_{1,n+1}$ is the overlap between a spin configuration of the set $\{1,2,...,n\}$ and and external
one that we enumerate as the $(n+1)$-st, and $c_{1,j}$ is the overlap between two generic spin configurations among
the $n$'s.

The proof ideas can be easily summarized by the study of three quantities
and their differences which encode the fluctuation properties of the spin
glass system:
\begin{equation}
\mathbb{E}\left[\omega(H_N^2)\right]  \; ,
\end{equation}
\begin{equation}
\mathbb{E}\left[\omega(H_N)^2\right]  \; ,
\end{equation}
\begin{equation}
\mathbb{E}\left[\omega(H_N)\right]^2  \; .
\end{equation}
The result is obtained by two bounds for constants $\epsilon^{(1)}_N$ and $\epsilon^{(2)}_N$
vanishing in the $N \to \infty$ limit:
\begin{itemize}
\item bound on averaged thermal fluctuations
\begin{equation}
\mathbb{E}\left[\omega(H_N^2)-\omega(H_N)^2\right] \le \epsilon_N^{(1)} N
\end{equation}
obtained by stochastic Stochastic Stability method (see \cite{AizenmanContucci}) by showing that
the addition of an independent term of order one to the Hamiltonian
is equivalent to a small change in temperature of the entire system:
\begin{equation}
\beta H_N(\sigma) \to \beta H_N(\sigma) + \sqrt{\frac{\lambda}{N}}{\tilde H}_N(\sigma)
\end{equation}
\begin{equation}
\beta \to \sqrt{\beta^2 +\frac{\lambda}{N}}
\end{equation}
\item bound on disorder fluctuations
\begin{equation}
U=\omega(H_N)
\end{equation}
\begin{equation}
\mathbb{E}(U^2)-\mathbb{E}(U)^2 \le \epsilon_N^{(2)} N\,  ,
\end{equation}
which is the self averaging of internal energy and can be proved from self averaging
of the free energy (with martingale methods or concentration of measures).
\end{itemize}
In the following sections we show how to use the previous ideas to obtain a stronger
result, namely the validity of the previous properties in distribution for quenched
probability measure of the Hamiltonian covariance.

The paper is organised as follows: Section 2 introduces the basic definitions to construct the 
class of models which satisfies thermodynamic stability. Section 3 shows the identities which 
can be derived from such properties. Section 4 deals with the rate of convergence in the thermodynamic limit.
Last section contains the discussion on the connection between those identities
and the Parisi's ultrametric property.

\section{Definitions and preliminary properties}

The model we chose to work with is the most general spin glass in $d$-dimension
and at the same time the closest to physical reality. Physical particles in fact, beside
interacting in pairs have always higher order interactions, i.e. they interact in triples, quadruples etc.
(see \cite{Ruelle}). Given a lattice in dimension $d$, for example $\Z^d$ we consider,
for each finite set $\Lambda \subset \Z^d$, an Hamiltonian of the form
\begin{equation}\label{h_alaruelle}
H_{\Lambda}(\sigma)\; = \;
- \sum_{X \subseteq \Lambda} J_{\Lambda,X} \sigma_X\, ,
\end{equation}
where $\sigma_X = \prod_{x \in X} \sigma_x$ and where all the random couplings $J_{\Lambda,X}$ are independent
centered Gaussian random variables with $\E[J_{\Lambda,X}^2]\, =\, \Delta_{\Lambda,X}^2$ for some nonnegative
constants $(\Delta_{\Lambda,X})_{X \subseteq \Lambda}$.

The thermodynamical properties of the previous models are encoded within the quenched
distribution of its normalized covariance
\begin{equation}\label{g_covariance}
c_{\Lambda}(\sigma,\tau)\;=\; \frac{1}{|\Lambda|}
 \E[H_{\Lambda}(\sigma) H_{\Lambda}(\tau)]\;=\;
\frac{1}{|\Lambda|}\sum_{X \subseteq \Lambda} \Delta^2_{\Lambda,X} \sigma_X \tau_X\, .
\end{equation}

The condition for existence of the thermodynamic limit (in the sense of Fisher)
called ``thermodynamic stability'' is:
\begin{equation}\label{thestability}
\sup_{\Lambda \subseteq \Z^d} \frac{1}{|\Lambda|}\, \sum_{X \subseteq \Lambda} \Delta^2_{\Lambda,X}\,  \le \, c\, <\, \infty\, ,
\end{equation}
see \cite{ContucciGiardina, ContucciGiardinaMonograph}.

In order to introduce the necessary language to illustrate our results we start by the following:

\begin{lemma}\label{lem:semialgebra1} Let $c_{\Lambda}$ and $c'_{\Lambda}$ be two normalised covariances
of Gaussian spin glasses satisfying the condition of thermodynamic stability. Then the same condition is satisfied
by the normalised covariance obtained through the operations below:

\begin{itemize}
\item $c_{\Lambda}+c_{\Lambda}'$, entry-wise addition
\item $x^2 c_{\Lambda}$ for each $x \in \R$, scalar multiplication
\item $c_{\Lambda} c_{\Lambda}'$, entry-wise multiplication.
\end{itemize}
\end{lemma}

\begin{proof}
We first observe that the three considered operations define new covariances, in particular
the last covariance is sometimes called the Schur product, Hadamard
product or Schur-Hadamard product and is semidefinite positive by a lemma of Schur.

The conditions of thermodynamic stability for $c_{\Lambda}$ and $c_{\Lambda}'$ are
$$
\sup_{\Lambda} c_{\Lambda}(\sigma,\sigma)\, <\, \infty\quad \text{ and } \quad
\sup_{\Lambda} c'_{\Lambda}(\sigma,\sigma)\, < \infty\, ,
$$
and immediately imply that
\begin{gather*}
\sup_{\Lambda} [c_{\Lambda}(\sigma,\sigma) + c'_{\Lambda}(\sigma,\sigma)]\, <\, \infty\, ,\\
\sup_{\Lambda} x^2 c_{\Lambda}(\sigma,\sigma)\, <\, \infty\quad \text{ and}\\
\sup_{\Lambda} c_{\Lambda}(\sigma,\sigma) c_{\Lambda}'(\sigma,\sigma)\, <\, \infty\, .
\end{gather*}
The explicit inversion formula from the covariance to the Hamiltonian can be seen from Chapter 2 of \cite{ContucciGiardinaMonograph}.
\end{proof}

The previous lemma says that the set of the thermodynamically stable covariances is closed under
the three operations defined.

By the previous lemma, starting from the Hamiltonian (\ref{h_alaruelle}) we can always construct a thermodynamically stable Hamiltonian, that we call complete Hamiltonian, defined by
\begin{equation}\label{h_MP}
{\bar H}_{\Lambda}(\sigma;\beta):=\sum_{p\geq1}(\sqrt{c})^{-p}\beta_pH^{(p)}_{\Lambda}(\sigma),
\end{equation}
where $H^{(1)}_{\Lambda}(\sigma)\equiv H_{\Lambda}(\sigma)$ is the Hamiltonian (\ref{h_alaruelle})
and each $p$-term in the sum has a normalized covariance
$$c^{(p)}_{\Lambda}(\sigma,\tau)=[c_{\Lambda}(\sigma,\tau)]^p$$
and the family of parameters $\beta=(\beta_p)_{p\geq1}$ is such that $\beta_p>0$ for every $p$ and fulfills the condition
$$
\sum_{p\geq1}\beta^2_p=\overline{c}<\infty
$$

A simple computation shows that the complete Hamiltonian has a covariance
$$ \bar {c} _{\Lambda}(\sigma,\tau)= \sum_{p\geq1}(c)^{-p}\beta^2_p [c_{\Lambda}(\sigma,\tau)]^p$$
and is thermodynamically stable with constant $\bar {c}$.
\newline

Consider $n$ copies of the configuration space denoted by  $\sigma^{1},\ldots, \sigma^{n}$ and,
for every bounded function $f:(\sigma^{1},\ldots, \sigma^{n})\rightarrow \mathbb{R}$, we call the random $n$-Gibbs state the following r.v.
\be\label{gibbs}
\Omega_{\Lambda,\beta}(f):=\sum_{\sigma^{1},\ldots, \sigma^{n}}
f(\sigma^{1},\ldots, \sigma^{n})\mathcal{G}_{\Lambda,\beta}(\sigma^{1})\ldots \mathcal{G}_{\Lambda,\beta}(\sigma^{n})
\ee
where
\be\label{gibbs2}
\mathcal{G}_{\Lambda,\beta}(\sigma):=\frac{\exp(-{\bar H}_{\Lambda}(\sigma;\beta))}{\sum_{\sigma}\exp(-{\bar H}_{\Lambda}(\sigma;\beta))}
\ee
is the random Gibbs measure. In the previous formula the dependence on the physical $\beta$ is
reabsorbed in the family of $\beta_p$'s.

We can define the quenched Gibbs state as
\be\label{gibbs3}
\langle f \rangle_{\Lambda,\beta}:=\E\Omega_{\Lambda,\beta}(f)
\ee

\section{Identities}
\begin{theorem}\label{thm:main}
The model defined by equation (\ref{h_MP}) satisfies with respect
to the covariance (\ref{g_covariance}) the following properties:
\newline

(i) It is stochastically stable in the strong sense, i.e for every power $p \in \N$ and for almost every $\beta_p$,
the following hold
\be\label{thm:SACI}
\lim_{\Lambda \nearrow \Z^d}
\Big\langle \sum_{\substack{j,k=1\\j\neq k}}^{n} f c^p_{j,k}
- 2 n f \sum_{k=1}^{n} c^p_{k,n+1} + n(n+1) f c^p_{n+1,n+2}
\Big\rangle_{\Lambda,\beta}\,
=\, 0\, ,
\ee
where for any number of replicas $\sigma^{(1)},\sigma^{(2)},\dots$, we denote by $c_{j,k}$
the quantity $c_{\Lambda}(\sigma^{(j)},\sigma^{(k)})$, and where
we assume that $f$ is a continuous function of all the variables
$c_{j,k}$ for $1\leq j<k\leq n$.
\newline

(ii) It
fulfills the Ghirlanda-Guerra identities (GG for short)
in distribution, i.e.
the following identities are verified for every $n\geq2$ and every function $f$ of $(c_{j,k})_{j,k=1}^{n}$ as above, and every power $p \in \N$ and for almost every $\beta_p$.
\be
\label{thm:EGGI}
\lim_{\Lambda \nearrow \Z^d}
\Big[\langle f c_{n+1,n+2}^p \rangle_{\Lambda,\beta}
- \frac{1}{n+1}\, \sum_{k=1}^{n} \langle f c_{k,n+1}^p \rangle_{\Lambda,\beta}
- \frac{1}{n+1}\, \langle f\rangle_{\Lambda,\beta}
\langle c_{1,2}^p \rangle_{\Lambda,\beta}\Big]\, =\, 0\, ,
\ee
Moreover, let
$$
\bar{p}(\beta):=\lim_{\Lambda \nearrow \Z^d}\frac{1}{|\Lambda|}\E[\log \sum_{\sigma}\exp(-\bar{H}_{\Lambda}(\sigma,\beta))]
$$
be the thermodynamical limit of the pressure.
\newline

If $\bar{p}(\beta)$ is differentiable in the $\beta_p$ "direction" at the point $\beta_p=a$ then
\newline

(iii) It
is \textbf{pointwise} stochastically stable in the strong sense, i.e for every power $p \in \N$ and in each point  $\beta_p=a$, (\ref{thm:SACI}) hold.
\newline

(iv) It
fulfills the Ghirlanda-Guerra identities
in distribution \textbf{pointwise}, i.e for every $n\geq2$ and for every power $p \in \N$ and in each point  $\beta_p=a$, (\ref{thm:EGGI}) holds.
\end{theorem}

\begin{remark}we notice  that since the function $\overline{p}(\beta)$ is convex in each $\beta_p$ then it's almost everywhere differentiable and then we have that $(iv)\Rightarrow (ii)$, but in the next sections we give an independent proof of $(ii)$.
\end{remark}

\textbf{Proof of Theorem \ref{thm:main}}: We fix an arbitrary $p\geq1$ and to lighten the notation we put:

$$\beta_p\rightarrow x$$
$$\sum_{k\neq p}(\sqrt{c})^{-k}\beta_kH^{(k)}_{\Lambda}(\sigma)\rightarrow
H_{\Lambda}(\sigma)$$
$$(\sqrt{c})^{-p}\beta_pH^{(p)}_{\Lambda}(\sigma)\rightarrow H'_{\Lambda}(\sigma)$$
Then the Hamiltonian defined in (\ref{h_MP}) becomes
$$\bar{H}_{\Lambda}(\sigma;\beta)\rightarrow H_{\Lambda}(\sigma;x)=H_{\Lambda}(\sigma)+xH'_{\Lambda}(\sigma)$$
To prove the theorem we recall a general result due to Panchenko \cite{PanchenkoMP}. Consider a general Hamiltonian of the type
$$H_{\Lambda}(\sigma;x)=H_{\Lambda}(\sigma)+xH'_{\Lambda}(\sigma)$$
where $x$ is a real parameter and the families $(H(\sigma))_{\sigma}$ and $(H'(\sigma))_{\sigma}$ are independent jointly Gaussian families of centered r.v.s. of the type
(\ref{h_alaruelle}).
Suppose that the Hamiltonian are thermodynamically stable in the sense of (\ref{thestability}), that is there exists a global constant $c$ such that

$$\E H_{\Lambda}(\sigma)^2\leq |\Lambda|c$$
\be\label{termosta}
\E H'_{\Lambda}(\sigma)^2\leq |\Lambda|c
\ee
Consider, the following basic quantities:

$$Z_{\Lambda}(x):=\sum_{\sigma}\exp (-H_{\Lambda}(\sigma;x))$$

$$p_{\Lambda}(x):=\frac{1}{|\Lambda|}\log Z_{\Lambda}(x)$$
\be\label{defpressure}
p(x):=\lim_{\Lambda \nearrow \Z^d}\E p_{\Lambda}(x)
\ee
Notice that existence of the limit in the last definition is ensured (see for example \cite{ContucciGiardinaMonograph}) by the conditions (\ref{termosta}). We define $\Omega_{\Lambda,x}(\;)$ and $\langle \; \rangle_{\Lambda,x}$ in the same way as in
(\ref{gibbs}) and  (\ref{gibbs3}).
\newline

In the previous setting we have the following lemma:
\begin{lemma} \label{almpoint}

If we denote  ${h'_{\Lambda}(\sigma)}:=\frac{1}{|\Lambda|}H'_{\Lambda}(\sigma)$ the Hamiltonian density, then we have that for every $\beta_1<\beta_2$

\be\label{ACaverage}
\lim_{\Lambda \nearrow \Z^d}\int^{\beta_2}_{\beta_1}\Big\langle \Big|h'_{\Lambda}(\sigma)- \Omega_{{\Lambda},a}\Big(h'(\sigma)\Big)\Big|\Big\rangle_{{\Lambda},a}da=0
\ee

\be\label{GGaverage}
\lim_{\Lambda \nearrow \Z^d}\int^{\beta_2}_{\beta_1}\Big\langle \Big|h'_{\Lambda}(\sigma)- \langle h'_{\Lambda}(\sigma)\rangle_{\Lambda,a}\Big|\Big\rangle_{\Lambda,a}da=0
\ee

On other hand, if we assume that $p(x)$
 is differentiable at $x = a$, then
\be\label{ACpoint}
\lim_{\Lambda \nearrow \Z^d}\Big\langle \Big|h'_{\Lambda}(\sigma)- \Omega_{{\Lambda},a}\Big(h'(\sigma)\Big)\Big|\Big\rangle_{{\Lambda},a}=0
\ee

\be\label{GGpoint}
\lim_{\Lambda \nearrow \Z^d}\Big\langle \Big|h'_{\Lambda}(\sigma)- \langle h'_{\Lambda}(\sigma)\rangle_{\Lambda,a}\Big|\Big\rangle_{\Lambda,a}=0
\ee
\end{lemma}

It is easy to check by a simple integration-by-parts and a uniform norm bound that the relations (\ref{ACaverage}), (\ref{GGaverage}), (\ref{ACpoint}), (\ref{GGpoint}) implies the propositions $i), ii), iii), iv)$ of Theorem \ref{thm:main}, respectively. We notice that the propositions $i), ii)$ are in almost every sense then in this case the proof of previous implication requires some elementary facts in measure theory which are explained in Remark \ref{almostevery}.\\

\textbf{Proof of Lemma \ref{almpoint}}: The strategy of the proof is to control all terms
by the following estimation, which is essentially contained in Chapter 12 of \cite{Talagrand2}.

\begin{proposition}\label{talestimate}
For every $b>0$ we have that
\be\label{ACtalestimate}
\Big\langle \Big|h'_{\Lambda}(\sigma)- \Omega_{{\Lambda},a}\Big(h'(\sigma)\Big)\Big|\Big\rangle_{{\Lambda},a}\leq\sqrt{\frac{2\E D_{\Lambda}(a,b)}{b|\Lambda|}}+8\E D_{\Lambda}(a,b)
\ee
\be\label{GGtalestimate}
\Big\langle \Big|\Omega_{{\Lambda},a}\Big(h'(\sigma)\Big)- \langle h'_{\Lambda}(\sigma)\rangle_{{\Lambda},a}\Big|\Big\rangle_{{\Lambda},a}\leq \E D_{\Lambda}(a,b)+\E W_{\Lambda}(a,b)
\ee
where

$$D_{\Lambda}(x,b):=p_{\Lambda}'(x+b)-p_{\Lambda}'(x-b)$$
$$W_{\Lambda}(x,b):=\frac{1}{b}\Big(|p_{\Lambda}(x+b)-\E p_{\Lambda}(x+b)|+|p_{\Lambda}(x-b)-\E p_{\Lambda}(x-b)|+|p_{\Lambda}(x)-\E p_{\Lambda}(x)|\Big)$$
\end{proposition}
\textbf{Proof of Proposition \ref{talestimate}, equation (\ref{GGtalestimate})}: the function $p_{\Lambda}(x)$ is convex. Thus for every $b>0$ we have that
$$p'_{\Lambda}(x)\leq \frac{p_{\Lambda}(x+b)-p_{\Lambda}(x)}{b}\leq W_{\Lambda}(x,b)+\frac{\E \Big(p_{\Lambda}(x+b)-p_{\Lambda}(x)\Big)}{b}\leq W_{\Lambda}(x,b)+\E p'_{\Lambda}(x+b)\, ,$$
and then
$$p'_{\Lambda}(x)-\E p'_{\Lambda}(x)\leq W_{\Lambda}(x,b)+\E p'_{\Lambda}(x+b)-\E p'_{\Lambda}(x)\leq W_{\Lambda}(x,b)+\E D_{\Lambda}(x,b)\, .$$

On other hand,
$$p'_{\Lambda}(x)\geq -\frac{p_{\Lambda}(x-b)+p_{\Lambda}(x)}{b}$$
and then, after the same manipulations, we get
$$p'_{\Lambda}(x)-\E p'_{\Lambda}(x)\geq -W_{\Lambda}(x,b)-\E D_{\Lambda}(x,b)$$

Combining the two previous inequalities we obtain the following bound
\be\label{boundself}
|p'_{\Lambda}(x)-\E p'_{\Lambda}(x)|\leq W_{\Lambda}(x,b)+\E D_{\Lambda}(x,b)
\ee
This bound give immediately (\ref{GGtalestimate}), since
$$\Big\langle \Big|\Omega_{{\Lambda},a}\Big(h'(\sigma)\Big)- \langle h'_{\Lambda}(\sigma)\rangle_{{\Lambda},a}\Big|\Big\rangle_{{\Lambda},a}=\E\Big(\Big|p'_{\Lambda}(a)-\E p'_{\Lambda}(a)\Big|\Big)\leq \E W_{\Lambda}(a,b)+\E D_{\Lambda}(a,b)\, .\quad \Box$$

The bound (\ref{ACtalestimate}) requires an extra work.
\begin{proposition}\label{psiestimate}
Consider the quantity
$$\psi_{\Lambda}(x):=\Omega_{\Lambda,x}\Big(\Big|h'_{\Lambda}(\sigma^{(1)})-h'_{\Lambda}(\sigma^{(2)})\Big|\Big)$$
then we have that
$$\psi^2_{\Lambda}(x)\leq\frac{4}{|\Lambda|}p''_{\Lambda}(x)$$
$$|\psi'_{\Lambda}(x)|\leq 8 p''_{\Lambda}(x)$$
\end{proposition}
\textbf{Proof of Proposition \ref{psiestimate}:} During this proof we define for sake of simplicity the quantity
$$V_{l}:=h'_{\Lambda}(\sigma^{(l)})-\Omega_{\Lambda,x}\Big(h'_{\Lambda}(\sigma)\Big)$$
then we have that, for every $l,m$

$$\Omega_{\Lambda,x}\Big(V_{l}\Big)=\Omega_{\Lambda,x}\Big(V_{m}\Big)=\Omega_{\Lambda,x}\Big(V_{1}\Big)=0$$
$$\Omega_{\Lambda,x}\Big(V_{l}V_{m}\Big)=\Omega_{\Lambda,x}\Big(V_{l}\Big)\Omega_{\Lambda,x}\Big(V_{m}\Big)=
\Omega^2_{\Lambda,x}\Big(V_{1}\Big)=0$$
$$p''_{\Lambda}(x)=|\Lambda|\Omega_{\Lambda,x}\Big(V^{2}_{1}\Big)$$

By Jensen's inequality and the previous equations we can obtain the first bound of the proposition, indeed
$$\psi^2_{\Lambda}(x)\leq\Omega_{\Lambda,x}\Big(\Big(h'_{\Lambda}(\sigma^{(1)})-h'_{\Lambda}(\sigma^{(2)})\Big)^2\Big)=
\Omega_{\Lambda,x}\Big(\Big(V_{1}-V_{2}\Big)^2\Big)\leq\frac{4}{|\Lambda|}p''_{\Lambda}(x)
$$
The second bound follow easily by Cauchy-Schwarz inequality, indeed
 $$|\psi'_{\Lambda}(x)|=|\Lambda|\Omega_{\Lambda,x}\Big(\Big|h'_{\Lambda}(\sigma^{(1)})-h'_{\Lambda}(\sigma^{(2)})\Big|
 \Big(h'_{\Lambda}(\sigma^{(1)})+h'_{\Lambda}(\sigma^{(2)})-2h'_{\Lambda}(\sigma^{(3)})\Big)\Big)\leq$$
 $$|\Lambda|\Omega_{\Lambda,x}\Big(\Big|V_{1}-V_{2}\Big|\cdot\Big|V_{1}-V_{3}+V_{2}-V_{3}\Big|\Big)
 \leq2|\Lambda|\Omega_{\Lambda,x}\Big(\Big|V_{1}-V_{2}\Big|\Big|V_{1}-V_{3}\Big|\Big)\leq
 2|\Lambda|\Omega_{\Lambda,x}\Big(\Big(V_{1}-V_{2}\Big)^2\Big)\leq8p''_{\Lambda}(x)\Box
$$

The last proposition can be used to obtain
\begin{proposition}\label{abspsiestimate}
 Given $b > 0$ then
$$\E|\psi_{\Lambda}(x)|\leq\sqrt{\frac{2\E D_{\Lambda}(a,b)}{b|\Lambda|}}+8\E D_{\Lambda}(a,b)$$
\end{proposition}
\textbf{Proof of Proposition \ref{abspsiestimate}:}
We observe that for $x-b\leq y\leq x + b$ we have
$$|\psi_{\Lambda}(x)-\psi_{\Lambda}(y)| ≤ \int^{x+b}_{x-b}|\psi'_{\Lambda}(t)| dt$$
and then
$$\Big|\int^{x+b}_{x-b}\Big(\psi_{\Lambda}(x)-\psi_{\Lambda}(y)\Big) dy\Big|\leq2b\int^{x+b}_{x-b}|\psi'_{\Lambda}(t)| dt$$
The identity
$$2b\psi_{\Lambda}(x) = \int^{x+b}_{x-b}\psi_{\Lambda}(y) dy+\int^{x+b}_{x-b}\Big(\psi_{\Lambda}(x)-\psi_{\Lambda}(y)\Big) dy
$$
implies that
$$|\psi_{\Lambda}(x)|=\frac{1}{2b}\Big|\int^{x+b}_{x-b}\psi_{\Lambda}(y) dy\Big|+\frac{1}{2b}\Big|\int^{x+b}_{x-b}\Big(\psi_{\Lambda}(x)-\psi_{\Lambda}(y)\Big) dy\Big|
$$
and then Proposition \ref{psiestimate} implies that
$$|\psi_{\Lambda}(x)|\leq\frac{1}{b|\Lambda|^{\frac{1}{2}}}\int^{x+b}_{x-b}\sqrt{p''_{\Lambda}(y)} dy +8\int^{x+b}_{x-b}p''_{\Lambda}(y) dy$$
We can use the Jensen inequality in the first term of the r.h.s of the previous relation to get
$$|\psi_{\Lambda}(x)|\leq\sqrt{\frac{2}{b|\Lambda|}}\Big(\int^{x+b}_{x-b}p''_{\Lambda}(y) dy\Big)^{\frac{1}{2}} +8\int^{x+b}_{x-b}p''_{\Lambda}(y) dy$$

To conclude, we take the expectation and using again the Jensen inequality and the obvious relation
$$\int^{x+b}_{x-b}p''_{\Lambda}(y) dy=p'_{\Lambda}(x+b)-p'_{\Lambda}(x-b)$$
and the proof is complete.$\Box$
\newline

\textbf{Proof of Proposition \ref{talestimate}, equation (\ref{ACtalestimate})}:
To obtain (\ref{ACtalestimate}) we simply observe that by Jensen inequality

$$\Big\langle \Big|h'_{\Lambda}(\sigma)- \Omega_{{\Lambda},a}\Big(h'(\sigma)\Big)\Big|\Big\rangle_{{\Lambda},a}\leq\E|\psi_{\Lambda}(a)|$$
and then by Proposition \ref{abspsiestimate} we get the desired result.$\Box$
\newline

Now we are able to prove Lemma \ref{almpoint}.

\textbf{Proof of Lemma \ref{almpoint}}:
We start to prove the equations (\ref{ACaverage}) and (\ref{ACpoint}) which give the stochastic stability.
\newline

First, using the convexity of the function $p_{\Lambda}(x)$ we can prove easily that

\be\label{Destimate}
D_{\Lambda}(a,b)\leq\
\frac{1}{b}\Big(p_{\Lambda}(a+2b)-p_{\Lambda}(a+b)+p_{\Lambda}(a-2b)-p_{\Lambda}(a-b)\Big)
\ee
It is easy to check that $\E p_{\Lambda}(x)$ is bounded for every $x,\Lambda$ and then for every $a,b,\Lambda$ we have that
$$\E D_{\Lambda}(a,b)\leq \frac{\bar{D}}{b}<\infty\, .$$.
Then  from (\ref{ACtalestimate}), using Fubini's Theorem, we get the following bound: for every $b,\beta_1<\beta_2$

$$\int^{\beta_2}_{\beta_1}\Big\langle \Big|h'_{\Lambda}(\sigma)- \Omega_{{\Lambda},a}\Big(h'(\sigma)\Big)\Big|\Big\rangle_{{\Lambda},a}da\leq\frac{(\beta_2-\beta_1)}{b}\sqrt{\frac{2 \bar{D}}{|\Lambda|}}+8\E \int^{\beta_2}_{\beta_1}D_{\Lambda}(a,b)da$$

and then

$$\int^{\beta_2}_{\beta_1}\Big\langle \Big|h'_{\Lambda}(\sigma)- \Omega_{{\Lambda},a}\Big(h'(\sigma)\Big)\Big|\Big\rangle_{{\Lambda},a}da\leq\frac{(\beta_2-\beta_1)}{b}\sqrt{\frac{2 \bar{D}}{|\Lambda|}}+8\E\Big(p_{\Lambda}(\beta_2+b)-p_{\Lambda}(\beta_1+b)
-p_{\Lambda}(\beta_2-b)+p_{\Lambda}(\beta_1-b)\Big) $$

Finally, we can use the $\lim_{b\rightarrow 0}\limsup_{\Lambda \nearrow \Z^d}$ and  the continuity of the function $p(x)$ to get (\ref{ACaverage}).
\newline

To prove (\ref{ACpoint}), we use the hypothesis of differentiability of the function $p(x)$ at the point $x=a$ to get from (\ref{Destimate}) the following
\begin{equation}\label{griffestimate}
\lim_{b\rightarrow 0}\limsup_{\Lambda \nearrow \Z^d}\E D_{\Lambda}(a,b)\leq \lim_{b\rightarrow 0} \frac{1}{b}\Big(p(a+2b)-p(a+b)+p(a-2b)-p(a-b)\Big)=0
\end{equation}
and then we can bypass the intermediate integration to obtain from (\ref{ACtalestimate}) the following bound: for every $a,b$
$$\Big\langle \Big|h'_{\Lambda}(\sigma)- \Omega_{{\Lambda},a}\Big(h'(\sigma)\Big)\Big|\Big\rangle_{{\Lambda},a}\leq\frac{1}{b}\sqrt{\frac{2 \bar{D}}{|\Lambda|}}+8\E D_{\Lambda}(a,b)
$$
Finally, we can use the $\lim_{b\rightarrow 0}\limsup_{\Lambda \nearrow \Z^d}$ and relation (\ref{griffestimate}) to get (\ref{ACaverage}).
\newline

Now, we are able to prove the equations (\ref{GGaverage}) and (\ref{GGpoint}) which give the $GG$-identities.
\newline

We already obtained the control of the quantity $D$, moreover a simple inspection shows that
the quantity $W_{\Lambda}(a,b)$ is strictly related to the self-averaging of $p_{\Lambda}(x)$. Then it's easy to check (see for example \cite{ContucciGiardina}) that the thermodynamic stability condition (\ref{termosta}) ensures that there exists a finite quantity $K(a,b,c)$ which
does not depend on $\Lambda$,  such that:

\begin{equation}\label{selfavestimate}
\E W_{\Lambda}(a,b)\leq K(a,b,c)\frac{1}{b|\Lambda|^{\frac{1}{2}}}
\end{equation}

From (\ref{GGtalestimate}) and (\ref{selfavestimate}), we can obtain as before the following bound: for every $b,\beta_1<\beta_2$

$$\int^{\beta_2}_{\beta_1}\Big\langle \Big|\Omega_{{\Lambda},a}\Big(h'(\sigma)\Big)- \langle h'_{\Lambda}(\sigma)\rangle_{\Lambda,a}\Big|\Big\rangle_{{\Lambda},a}da\leq \E\Big(p_{\Lambda}(\beta_2+b)-p_{\Lambda}(\beta_1+b)
-p_{\Lambda}(\beta_2-b)+p_{\Lambda}(\beta_1-b)\Big)
+\frac{\overline{K}(\beta_1,\beta_2,b,c)}{b|\Lambda|^{\frac{1}{2}}}
$$
where we have set $\overline{K}(\beta_1,\beta_2,b,c):=\int^{\beta_2}_{\beta_1}K(a,b,c)da$ to lighten the notation.
\newline
Finally, we can use the $\lim_{b\rightarrow 0}\limsup_{\Lambda \nearrow \Z^d}$ and  the continuity of the function $p(x)$ to get (\ref{GGaverage}).
\newline

Like in the previous case, the hypothesis of differentiability allow us to bypass the intermediate integration and then, from (\ref{GGtalestimate}), we have that, for every $a,b$

$$\Big\langle \Big|\Omega_{{\Lambda},a}\Big(h'(\sigma)\Big)- \langle h'_{\Lambda}(\sigma)\rangle_{\Lambda,a}\Big|\Big\rangle_{{\Lambda},a}\leq \E D_{\Lambda}(a,b)
+K(a,b,c)\frac{1}{b|\Lambda|^{\frac{1}{2}}}
$$
Finally, we can use the $\lim_{b\rightarrow 0}\limsup_{\Lambda \nearrow \Z^d}$ and relation (\ref{griffestimate}) to get (\ref{GGaverage}).$\Box$
\newline

\section{Rate of convergence}

We outline in this section a sharper version of a theorem that appears in \cite{Talagrand2} and prove it with more elementary methods
for the benefit of the reader, following the approach developed in \cite{AizenmanContucci,ContucciGiardina}.

As in Section 3, we consider
$$
H_{\Lambda}(\sigma;x)\, =\, H_{\Lambda}(\s) + x H_{\Lambda}'(\s)\, ,
$$
where $H_{\Lambda}$ and $H_{\Lambda}'$ are independent, and defined as in (2).
The main theorem in this section follows:
\begin{theorem}
\label{thm:rate}
(a)
Writing $h_{\Lambda}' = |\Lambda|^{-1} H_{\Lambda}'$, as before,
\begin{equation}
\label{eq:DesideratumA}
\int_{x_1}^{x_2} \left\langle \left|h_{\Lambda}'(\sigma) - \Omega_{\Lambda,x}\big(h_{\Lambda}'(\sigma)\Big)\right|^2 \right\rangle_{\Lambda,x}\, dx\,
\leq\, \frac{2 (|x_1|+|x_2|) \overline{c}_{\Lambda}'}{|\Lambda|}\, ,
\end{equation}
where $c_{\Lambda}'(\s,\s')$ is defined to be $|\Lambda|^{-1} |\E[H_{\Lambda}'(\sigma)H_{\Lambda}'(\sigma')]|$, and
$\overline{c}_{\Lambda}' \stackrel{\mathrm{def}}{:=}\, \max_{\sigma,\sigma'} |c_{\Lambda}'(\s,\s')|$.\\
(b) For any $x_1<x_2$
\begin{equation}
\label{eq:DesideratumB}
\frac{1}{x_2-x_1}\, \int_{x_1}^{x_2} \E\left[\left|\Omega_{\Lambda,x}\left(h'_{\Lambda}\right) -
\left\langle h'_{\Lambda} \right\rangle_{\Lambda,x}\right|\right]\, dx\,
\leq\,
4\left(|x_1|+|x_2|+\frac{2 \sqrt{x_2-x_1}}{(\overline{c}_{\Lambda}' |\Lambda|)^{1/4}}\right)\, \frac{(\overline{c}_{\Lambda}')^{3/4}}{\sqrt{x_2-x_1}\, |\Lambda|^{1/4}}\, .
\end{equation}
\end{theorem}

\subsection{Application: Distributional Stochastic Stability via Perturbations}

The quantitative version of the Ghirlanda-Guerra identities follows from this.
\begin{corollary}
\label{cor:QuantGG}
Suppose $c_{\Lambda}(\s,\s) = \overline{c}_{\Lambda}'$ for every $\s$. Then for every non-random function of $n$ replicas,
$\Phi_{\Lambda}^{(n)}(\s^{(1)},\dots,\s^{(n)})$, with maximum norm at most 1,
we have the conditional expectation formula for one additional replica
\begin{align*}
\frac{1}{x_2^2-x_1^2}\, \int_{x_1}^{x_2} \left\langle
\frac{c_{\Lambda}'(\s^{(1)},\s^{(n+1)})}{\overline{c}_{\Lambda}'}\,
\Phi^{(n)}(\s^{(1)},\dots,\s^{(n)}) \right\rangle_{\Lambda,x}\, d(x^2)\\
&\hspace{-6cm}=\, \frac{1}{x_2^2-x_1^2}\, \int_{x_1}^{x_2} \left\langle
\gamma_{\Lambda}^{(n)}(\s^{(1)},\dots,\s^{(n)})
\Phi^{(n)}(\s^{(1)},\dots,\s^{(n)}) \right\rangle_{\Lambda,x}\, d(x^2)
+ \operatorname{Rem}
\end{align*}
for every pair $0\leq x_1<x_2$,
where
$$
\gamma_{\Lambda}^{(n)}(\s^{(1)},\dots,\s^{(n)})\,
=\, \frac{1}{n}\, \left[\sum_{k=2}^{n} \frac{c_{\Lambda}'(\s^{(1)},\s^{(k)})}{\overline{c}_{\Lambda}'}
+
\left\langle \frac{c_{\Lambda}'(\s,\s')}{\overline{c}_{\Lambda}'} \right\rangle_{\Lambda,x}\right]\, ,
$$
and the remainder satisfies the bound
\begin{equation}
\label{ineq:RemBound}
n \left|\operatorname{Rem}\right|\,
\leq\, \left[\frac{8}{x_1+x_2} +\frac{2^{3/2}}{\sqrt{x_2^2-x_1^2}}\right] \delta_{\Lambda}^{1/2}
+ \frac{4}{\sqrt{x_2-x_1}}\, \delta_{\Lambda}^{1/4}\, ,
\end{equation}
where $\delta_{\Lambda}$ is a small parameter $\delta_{\Lambda} = 1/(|\Lambda| \overline{c}_{\Lambda})$.
\end{corollary}
\begin{proof}
Let us define
$$
\widehat{\operatorname{Rem}}\, =\,
\frac{1}{x_2-x_1}\, \int_{x_1}^{x_2}
\left(
\left\langle h_{\Lambda}'(\s^{(1)}) \cdot \Phi_{\Lambda}^{(n)}(\s^{(1)},\dots,\s^{(n)}) \right\rangle_{\Lambda,x}
- \left\langle h_{\Lambda}' \right\rangle_{\Lambda,x}
\left\langle \Phi_{\Lambda}^{(n)}(\s^{(1)},\dots,\s^{(n)}) \right\rangle_{\Lambda,x}\right)\, dx\, .
$$
By the triangle inequality and Cauchy-Schwarz, and the fact that $\|\Phi_{\Lambda}^{(n)}\|_{\infty}\leq 1$,
we know that $|\widehat{\operatorname{Rem}}|$ is bounded by
$$
\left(\frac{1}{x_2-x_1}\, \int_{x_1}^{x_2} \left\langle \left|h_{\Lambda}'(\sigma) - \Omega_{\Lambda,x}\big(h_{\Lambda}'(\sigma)\Big)\right|^2 \right\rangle_{\Lambda,x}\, dx\right)^{1/2}
$$
plus
$$
\frac{1}{x_2-x_1}\, \int_{x_1}^{x_2} \E\left[\left|\Omega_{\Lambda,x}\left(h'_{\Lambda}\right) -
\left\langle h'_{\Lambda} \right\rangle_{\Lambda,x}\right|\right]\, dx\, .
$$
Using equations (\ref{eq:DesideratumA}) and (\ref{eq:DesideratumB}),
this bound is at most $\overline{c}'_{\Lambda}(x_1+x_2)/2$ times the right hand side of (\ref{ineq:RemBound}).
In other words, we have an upper bound on $\widehat{\operatorname{Rem}}$
which is at most $\overline{c}'_{\Lambda}(x_1+x_2)/(2n)$ times the bound we claimed
for $\operatorname{Rem}$.

But Gaussian integration by parts implies that
\begin{align*}
-\left\langle
h_{\Lambda}'(\s^{(1)})
\Phi_{\Lambda}^{(n)}(\s^{(1)},\dots,\s^{(n)})\right\rangle_{\Lambda,x}\,
&=\, x \sum_{k=1}^{n} \left\langle c_{\Lambda}'(\s^{(1)},\s^{(k)}) \cdot
\Phi_{\Lambda}^{(n)}(\s^{(1)},\dots,\s^{(n)})
\right\rangle_{\Lambda,x}\\
&\qquad
- n x \left\langle c_{\Lambda}'(\s^{(1)},\s^{(n+1)}) \cdot \Phi_{\Lambda}^{(n)}(\s^{(1)},\dots,\s^{(n)})
\right\rangle_{\Lambda,x}\, .
\end{align*}
See for example Lemma \ref{lem:FirstDeriv} for a similar calculation carried out in more detail.
A special case of this formula, obtained by setting $\Phi^{(n)}\equiv 1$, also gives
$$
-\left\langle h_{\Lambda}' \right\rangle_{\Lambda,x}\,
=\,  x \left[ \langle c_{\Lambda}'(\s,\s) \rangle_{\Lambda,x}
- \langle c_{\Lambda}'(\s,\s') \rangle_{\Lambda,x} \right]\, .
$$
If we combine these two formulas, this allows one to rewrite $\widehat{\operatorname{Rem}}$.
If we assume that $c_{\Lambda}'(\s,\s)$ is constant as a function of $\s$, meaning there is a constant
diagonal covariance, then the $k=1$ term of the first formula cancels with the $c_{\Lambda}'(\s,\s)$
in the second formula.
Therefore, we get
$$
\widehat{\operatorname{Rem}}\, =\, \frac{\overline{c}_{\Lambda}'}{n(x_2-x_1)}\,
\int_{x_1}^{x_2} \left\langle
\left[
\frac{c_{\Lambda}'(\s^{(1)},\s^{(n+1)})}{\overline{c}_{\Lambda}'}
- \gamma_{\Lambda}^{(n)}(\s^{(1)},\dots,\s^{(n)})\right]\cdot
\Phi^{(n)}(\s^{(1)},\cdots,\s^{(n)}) \right\rangle_{\Lambda,x}\, x\, dx\, .
$$
Note that the measure $x\, dx$ is $\frac{1}{2} \cdot d(x^2)$, writing the Riemann-Stieltjes differential form
$d(x^2) = 2x\, dx$.
Since
$$
\int_{x_1}^{x_2} d(x^2) = x_2^2-x_1^2=(x_2-x_1)(x_1+x_2)\, ,
$$
we also divide by an appropriate normalization $\frac{1}{2}(x_1+x_2)$ times $\overline{c}_{\Lambda}'/n$ to get the  bound for $\operatorname{Rem}$
from the bound on $\widehat{\operatorname{Rem}}$, which gives the result.
\end{proof}

As an application of this result, consider the following scenario. Suppose that for each $\Lambda$,
there is a given Hamiltonian $H_{\Lambda}^*$ with covariance
$$
\E[H_{\Lambda}^{*}(\s) H_{\Lambda}^{*}(\s')]\, =\, |\Lambda| c_{\Lambda}^{*}(\s,\s')\, ,
$$
where we assume that $c_{\Lambda}^{*}(\s,\s) = \overline{c}_{\Lambda}^*$
for all $\s$, and we assume that $\overline{c}^* = \sup_{\Lambda} \overline{c}_{\Lambda}^*$
is finite, in order to satisfy thermodynamic stability.

By Lemma \ref{lem:semialgebra1} we know that we may construct IID Gaussian centered Hamiltonians
$H_{\Lambda}^{(p)}(\s)$, for $p=1,2,\dots$, which are independent of $H_{\Lambda}^*$ and such
$$
\E[H^{(p)}_{\Lambda}(\s) H^{(p)}_{\Lambda}(\s')]\, =\, |\Lambda| \left[c_{\Lambda}^{*}(\s,\s')\right]^p\, .
$$
For each $\epsilon>0$ and a real sequence $\boldsymbol{x} = (x_1,x_2,\dots)$
we define the perturbed Hamiltonian
$$
H_{\Lambda}(\sigma;\boldsymbol{x})\, =\, H_{\Lambda}^{*}(\sigma)
+ \sum_{p=1}^{\infty} \frac{x_p}{p [\overline{c}^{*}]^{p/2}}\, H_{\Lambda}^{(p)}(\sigma)\, .
$$
We denote by $\langle \cdots \rangle_{\Lambda,\boldsymbol{x}}$ the quenched
multi-replica equilibrium measure with respect to $H^{\epsilon}_{\Lambda}(\sigma;\boldsymbol{x})$.
Then we may prove the following corollary.

\begin{corollary}\label{co43} Suppose that the sequence $(\epsilon_{\Lambda})$
satisfies
$\lim_{|\Lambda| \to \infty} |\Lambda| \epsilon_{\Lambda}^2 = \infty$.
Let $\boldsymbol{X} = (X_1,X_2,\dots)$ be an IID sequence of random variables, each uniformly
distributed on $[0,1]$, all of which are independent of $H_{\Lambda}^{*}$ and $H_{\Lambda}^{(p)}$
for $p=1,2,\dots$ and all $\Lambda$. Then for almost every choice of $\boldsymbol{X}$
we have stochastic stability in distribution:
for each $n,p \in \{1,2,\dots\}$,
$$
\lim_{|\Lambda|\to\infty}
{\max_{\Phi^{(n)}}}'
\left\langle \left([c_{\Lambda}^*(\s^{(1)},\s^{(n+1)})]^p
-
\frac{1}{n}\, \sum_{k=2}^{n}[c_{\Lambda}^*(\s^{(1)},\s^{(k)})]^p
- \frac{1}{n}\, \left\langle[c_{\Lambda}^*(\s,\s')]^p \right\rangle_{\Lambda,\epsilon_{\Lambda} \boldsymbol{X}}
\right)
\Phi^{(n)} \right\rangle_{\Lambda,\epsilon_{\Lambda}\boldsymbol{X}}\,
=\, 0\, ,
$$
where $\max'_{\Phi^{(n)}}$ is the maximum over all non-random functions
$\Phi^{(n)}(\s^{(1)},\dots,\s^{(n)})$ satisfying $\|\Phi^{(n)}\|\leq 1$.
\end{corollary}
\begin{proof}
In order to prove this, for a given $p$, we merely split up the Hamiltonian:
$$
H_{\Lambda}(\sigma;\epsilon\boldsymbol{x})\, =\, H_{\Lambda}(\s) + x_p H_{\Lambda}'(\s)\, =:\, H_{\Lambda}(\s;x_p)\, ,
$$
where
$$
H_{\Lambda}(\sigma)\, =\, H^*_{\Lambda}(\sigma) + \epsilon
\sum_{\substack{k=1\\k\neq p}}^{\infty} \frac{x_k}{k [\overline{c}^*]^{k/2}}\, H_{\Lambda}^{(k)}(\s)\, ,
$$
and
$$
H'_{\Lambda}(\sigma)\, =\, \frac{\epsilon}{p [\overline{c}^*]^{p/2}}\, H_{\Lambda}^{(p)}(\s)\, .
$$
With this definition, we have $c'_{\Lambda}(\s,\s') = \frac{\epsilon^2}{p^2}\, [c^*_{\Lambda}(\s,\s')]^p$ but the constant prefactor has been explicitly taken into consideration
in Corollary \ref{cor:QuantGG} by normalizing by $\overline{c}'_{\Lambda}$.
It does enter into the definition of the remainder in (\ref{ineq:RemBound})
through the small parameter which is
$$
\delta_{\Lambda}\, =\, \frac{1}{|\Lambda|\overline{c}'_{\Lambda}}\, \leq\, \frac{p^2}{\epsilon_{\Lambda}^2 |\Lambda}\, .
$$
That is why we required $|\Lambda| \epsilon_{\Lambda}^2 \to \infty$,
because this guarantees $\delta_{\Lambda} \to 0$, as is needed.
\end{proof}
\begin{remark}\label{almostevery}
Note that technically what we proved is that for any open interval for $x_p$ in $[0,1]$, if we average the distributional
stochastic stability equation over that interval, when integrated against the measure $x\, dx$, then we obtain zero in the limit.
On the other hand, the quantity in question is
$$
\left\langle \left([c_{\Lambda}^*(\s^{(1)},\s^{(n+1)})]^p
-
\frac{1}{n}\, \sum_{k=2}^{n}[c_{\Lambda}^*(\s^{(1)},\s^{(k)})]^p
- \frac{1}{n}\, \left\langle[c_{\Lambda}^*(\s,\s')]^p \right\rangle_{\Lambda,\epsilon_{\Lambda} \boldsymbol{X}}
\right)
\Phi^{(n)} \right\rangle_{\Lambda,\epsilon_{\Lambda}\boldsymbol{X}}
$$
and this is bounded for every $x$
by $2 [\overline{c}_{\Lambda}^*]^p$.
Then by standard arguments from measure theory, we may conclude that for almost every
choice of $x$ with respect to the measure $d\mu(x) = 2x\, dx = d(x^2)$,
the quantity is also zero. But this measure is equivalent to Lebesgue measure in
the sense that they are mutually absolutely continuous with respect to each other.
So the notions of measure zero sets are the same.
\end{remark}
This means that letting $X_p$ be random, then for almost every $X_p$ we have the stochastic
stability formula for the $p$th power of the overlap.
But, firstly, we note that we may rigorously take an infinite number of IID uniform random variables
$\boldsymbol{X}=(X_1,X_2,\dots)$ by Kolmogorov's principle,
and secondly that the measure is precisely the product measure for all the $X_p$'s.
Therefore, knowing that for each $X_p$ we have the stochastic stability condition for almost
every $X_p$, by definition, this means we have the stochastic stability condition for all $p$
for almost every $\boldsymbol{X}$.

\subsection{Proof of Theorem \ref{thm:rate}}
In the proof of Theorem \ref{thm:rate} we will condition on $H_{\Lambda}(\sigma)$ in order to eliminate the need to consider
it as random.
But we will do this implicitly.
If desired, simply interpret all expectations as conditional expectations, conditioning
on $H_{\Lambda}(\sigma)$.

The proof will be obtained by combining several lemmas.
First, we note that by usual calculations as in elementary statistical mechanics,
\begin{equation}
\label{eq:PressureDerivative}
\frac{d}{dx}\, p_{\Lambda}(x)\, =\,
- \Omega_{\Lambda,x}(h_{\Lambda}')\, ,
\end{equation}
as has been used already in Section 3.
Moreover, by performing Gaussian integration by parts, we may deduce this:

\begin{lemma}
\label{lem:FirstDeriv}
For any $x$
\begin{equation}
\label{eq:pressureIBP}
\frac{d}{dx}\, \E[p_{\Lambda}(x)]\,
=\, -\langle h_{\Lambda}' \rangle_{\Lambda,x}\,
=\, \frac{x}{2}\, \left\langle c_{\Lambda}'\left(\s,\s\right) + c_{\Lambda}'\left(\s',\s'\right)
 - 2 c_{\Lambda}'\left(\s,\s'\right) \right\rangle_{\Lambda,x}\, .
\end{equation}
\end{lemma}
\begin{proof}
Since $\Lambda$ is finite, the derivative exists and we may write
$$
\frac{d}{dx}\, \E[p_{\Lambda}(x)]\,
=\,  \E[\Omega_{\Lambda,x}( h_{\Lambda}' )]\,
$$
by using (\ref{eq:PressureDerivative}). Recall the definition of $\Omega_{\Lambda,x}$
as well as $\mathcal{G}_{\Lambda,x}$ from equations (\ref{gibbs}),
(\ref{gibbs2}) and (\ref{gibbs3}).
Then Wick's rule gives
\begin{align*}
\E[\Omega_{\Lambda,x}( h_{\Lambda}' )]\,
&=\, \E\left[ \sum_{\s} h_{\Lambda}'(\s) \frac{\exp\left(-H_{\Lambda}(\sigma;x)\right)}{Z_{\Lambda}(x)}\right]\\
&=\, \sum_{\s} \E[h_{\Lambda}'(s) H_{\Lambda}'(\s)] \E\left[\frac{(\partial/\partial H_{\Lambda}'(\s)) \exp\left(-H_{\Lambda}(\s;x)\right)}{Z_{\Lambda}(x)}\right]\\
&\qquad
- \sum_{\s} \sum_{\s'}\E[h_{\Lambda}'(\s) H_{\Lambda}'(\s')]
\E\left[\frac{\exp(-H_{\Lambda}(\s;x))}{Z_{\Lambda}(x)^2} \cdot \frac{\partial}{\partial H_{\Lambda}'(\s')}\, \exp(-H_{\Lambda}(\s';x))\right]\, ,
\end{align*}
Since $H_{\Lambda}(\sigma;x) = H_{\Lambda}(\s) + x H'_{\Lambda}(\s)$,
this means $(\partial/\partial H_{\Lambda}'(\s)) \exp(-H_{\Lambda}(\s;x)) = -x \exp(H_{\Lambda}(\s))$.
Using this and the fact that $\E[h_{\Lambda}'(\s) H_{\Lambda}'(\s')] = |\Lambda|^{-1}
\E[H_{\Lambda}'(\s) H_{\Lambda}'(\s')] = c_{\Lambda}'(\s,\s')$, this gives the result.
\end{proof}

\begin{corollary}
\label{cor:FirstDerivCorollary}
For any $x$
$$
\left|\langle h_{\Lambda}' \rangle_{\Lambda,x}\right|\,
\leq\, 2|x| \overline{c}_{\Lambda}'\, .
$$
\end{corollary}
\begin{proof}
This follows from (\ref{eq:pressureIBP}) and a uniform bound using the definition
$\overline{c}_{\Lambda}' \stackrel{\mathrm{def}}{:=} \max_{\sigma,\sigma'} |c_{\Lambda}(\s,\s')|$.
\end{proof}

With this, we can prove the first part of the theorem.

\begin{proofof}{\bf Proof of Theorem \ref{thm:rate}, part (a):}
Another statistical mechanics  calculation following (\ref{eq:PressureDerivative})
is
$$
\frac{d^2}{dx^2}\, p_{\Lambda}(x)\,
=\, |\Lambda| (\Omega_{\Lambda,x}([h_{\Lambda}']^2)
- [\Omega_{\Lambda,x}(h_{\Lambda}')]^2)\, .
$$
So, integrating and taking expectations, we have
$$
\int_{x_1}^{x_2} \E\left[\Omega_{\Lambda,x}\left( \left|h_{\Lambda}'(\sigma) - \Omega_{\Lambda,x}\big(h_{\Lambda}'(\sigma)\Big)\right|^2\right)\right]\, dx\,
=\, \frac{1}{|\Lambda|} \cdot \frac{d}{dx}\, \E[p_{\Lambda}(x)]\bigg|_{x_1}^{x_2}\, .
$$
Recall that $\langle \cdots \rangle_{\Lambda,x} = \E[\Omega_{\Lambda,x}(\cdots)]$.
But by (\ref{eq:pressureIBP}) and Corollary \ref{cor:FirstDerivCorollary} this leads directly to (\ref{eq:DesideratumA}).
\end{proofof}

In order to obtain  the proof of Theorem \ref{thm:rate}, part (b), we will use concentration of measure.
Our main goal will be to obtain a bound on the random fluctuations
of the quantity
$\Omega_{\Lambda,x}\left(h_{\Lambda}'\right) - \langle h_{\Lambda}'\rangle_{\Lambda,x}$.
We start by quoting a result which was proved in \cite{GiardinaStarr}:
\begin{equation}
\label{eq:GS-bound}
\forall t>0\, ,\qquad
\P\left(\left|p_{\Lambda}(x) - \E[p_{\Lambda}(x)]\right|
\geq \frac{|x| \sqrt{\overline{c}_{\Lambda}'}}{\sqrt{|\Lambda|}}\, t \right)\, \leq\, 2 e^{-t^2/2}\, ,\
a.s.,
\end{equation}
where recall that $\overline{c}_{\Lambda}'$ was defined in the statement of Theorem \ref{thm:rate}, part (b).

We now claim that the following result may be proved using this and previous results:
\begin{lemma}
\label{lem:DerivCOM}
For any $x$, and for any $\epsilon>0$, we have
\begin{equation}
\label{eq:FirstVariation}
\E\left[\left|\Omega_{\Lambda,x}\left(h'_{\Lambda}\right) -
\left\langle h'_{\Lambda} \right\rangle_{\Lambda,x}\right|\right]\,
\leq\,
\frac{4(|x|+\epsilon) \sqrt{2\overline{c}_\Lambda'}}{\epsilon \sqrt{\pi |\Lambda|}}
+ \epsilon \int_{-1}^{1} (1 - |z|) \E\left[ \frac{d^2 p_{\Lambda}}{d x^2}(x+\epsilon z)\right]\, dz\, .
\end{equation}
\end{lemma}
\begin{proof}
By (\ref{eq:PressureDerivative}) and Taylor's theorem
\begin{align*}
-\Omega_{\Lambda,x}(h'_{\Lambda})\,
&=\,
\frac{p_{\Lambda}(x+\epsilon) - p_{\Lambda}(x-\epsilon)}{2\epsilon}
+ \frac{1}{2}\, \int_{-1}^{1} \left(\frac{dp_{\Lambda}}{dx}(x) - \frac{dp_{\Lambda}}{dx}(x+\epsilon y)\right)\, dy\\
&=\,
\frac{p_{\Lambda}(x+\epsilon) - p_{\Lambda}(x-\epsilon)}{2\epsilon}\\
&\qquad
+ \epsilon\, \int_{-1}^{1} \left(\int_{-1}^{1}
\frac{\boldsymbol{1}_{\R_-}(y) \boldsymbol{1}_{(y,0)}(z)
- \boldsymbol{1}_{\R_+}(y) \boldsymbol{1}_{(0,y)}(z)}{2} \cdot
\frac{d^2p_{\Lambda}}{dx^2}(x+\epsilon z)\, dz\right)\, dy\\
&=\,
\frac{p_{\Lambda}(x+\epsilon) - p_{\Lambda}(x-\epsilon)}{2\epsilon}
+ \frac{\epsilon}{2}\, \int_{-1}^{1} g(z) \frac{d^2 p_{\Lambda}}{d x^2}(x+\epsilon z)\, dz\, ,
\end{align*}
where we intechanged the order of integration and
$$
g(z)\, =\, \begin{cases} z-1 & \text{ for $0<z<1$,}\\
1+z & \text{ for $-1<z<0$.}
\end{cases}
$$
For each $z \in [-1,1]$, let us define the random variable
$$
Z_{\epsilon}(z)\, =\, \frac{d^2 p_{\Lambda}}{d x^2}(x+\epsilon z)\, ,
$$
which is nonnegative.
Let us also define
$$
Y_{\epsilon}\, =\, \frac{\epsilon}{2}\, \int_{-1}^{1} g(z) Z_{\epsilon}(z)\, dz\, .
$$
Then we obtain
$$
\Omega_{\Lambda,x}(h'_{\Lambda})
- \E[\Omega_{\Lambda,x}(h'_{\Lambda})]\,
=\,
\frac{(p_{\Lambda}(x+\epsilon) - \E[p_{\Lambda}(x+\epsilon)]) - (p_{\Lambda}(x-\epsilon)-\E[p_{\Lambda}(x-\epsilon)]}{2\epsilon}
+ Y_{\epsilon} - \E[Y_{\epsilon}]\, .
$$
We may use equation
(\ref{eq:GS-bound}) and the general subset bound,
$\P(A\cup B) \leq \P(A) + \P(B)$, to obtain this:
$$
\Omega_{\Lambda,x}\left(h'_{\Lambda}\right)
- \E[\Omega_{\Lambda,x}(h'_{\Lambda})]\,
\leq\, \frac{(|x|+\epsilon) \sqrt{\overline{c}_\Lambda'}}{\epsilon \sqrt{|\Lambda|}}\, t
+ Y_{\epsilon} - \E[Y_{\epsilon}] \quad \text{ with probability $p \geq 1 - 4 e^{-t^2/2}$.}
$$
We have a similar statement for the lower bound. Therefore, again by the subset bound,
\begin{equation}
\label{eq:DerivCOMpenultimate}
\P\left(\left|\Omega_{\Lambda,x}\left(h'_{\Lambda}\right) -
\left\langle h'_{\Lambda} \right\rangle_{\Lambda,x}\right|
\geq\, \frac{(|x|+\epsilon) \sqrt{\overline{c}_\Lambda'}}{\epsilon \sqrt{|\Lambda|}}\, t
+ |Y_{\epsilon} - \E[Y_{\epsilon}]| \right)\, \leq\, 8 e^{-t^2/2}\, .
\end{equation}
Recall that $\langle h'_{\Lambda} \rangle_{\Lambda,x} =  \E[\Omega_{\Lambda,x}(h'_{\Lambda})]$
by definition.

So, using the fact that for any integrable random variable $X$, $\E[X] \leq \int_0^{\infty} \P(X\geq t)\, dt$,
we obtain
\begin{equation}
\label{ineq:Erf}
\E\left[\left|\Omega_{\Lambda,x}\left(h'_{\Lambda}\right) -
\left\langle h'_{\Lambda} \right\rangle_{\Lambda,x}\right|\right]\,
\leq\, \E\left[|Y_{\epsilon} - \E[Y_{\epsilon}]|\right] +  \frac{8(|x|+\epsilon) \sqrt{\overline{c}_\Lambda'}}{\epsilon \sqrt{|\Lambda|}} \int_0^{\infty} e^{-t^2/2}\, dt\, .
\end{equation}

But now note that
$$
Y_{\epsilon} - \E[Y_{\epsilon}]\,
=\, \frac{\epsilon}{2}\, \int_{-1}^{1} g(z) \left(Z_{\epsilon}(z)- \E\left[Z_{\epsilon}(z)\right]\right)\, dz\, ,
$$
which implies that
$$
\E\left[\left|Y_{\epsilon} - \E[Y_{\epsilon}]\right|\right]\,
\leq\,
 \frac{\epsilon}{2}\, \int_{-1}^{1} |g(z)|\E\left[\left|Z_{\epsilon}(z)- \E\left[Z_{\epsilon}(z)\right]\right|\right]\, dz\, .
$$
But since $Z_{\epsilon}(z)$ is nonnegative, almost surely, this means that
$$
\E\left[\left|Z_{\epsilon}(z)- \E\left[Z_{\epsilon}(z)\right]\right|\right]\,
\leq\, \E\left[\left|Z_{\epsilon}(z)\right|\right]+ \left|\E\left[Z_{\epsilon}(z)\right]\right|\,
=\, 2 \E\left[Z_{\epsilon}(z)\right]\, .
$$
Combining this bound with (\ref{ineq:Erf}) gives the desired inequality (\ref{eq:FirstVariation}).
\end{proof}

With this, we can prove the remainder of the theorem.

\begin{proofof}{\bf Proof of Theorem \ref{thm:rate}, part (b):}
For any $x_1<x_2$, integrating (\ref{eq:FirstVariation}) and dividing by the length of the interval we obtain
\begin{equation}
\label{eq:Penultimate}
\begin{split}
\frac{1}{x_2-x_1}\, \int_{x_1}^{x_2} \E\left[\left|\Omega_{\Lambda,x}\left(h'_{\Lambda}\right) -
\left\langle h'_{\Lambda} \right\rangle_{\Lambda,x}\right|\right]\, dx\,
&\leq\, \frac{2(|x_1|+|x_2|+2\epsilon) \sqrt{2\overline{c}_\Lambda'}}{\epsilon \sqrt{\pi |\Lambda|}}\\
&\hspace{-3cm}
+ \frac{\epsilon}{x_2-x_1}\, \int_{-1}^{1} (1 - |z|)
\left(\int_{x_1}^{x_2} \E\left[\frac{d^2 p_{\Lambda}}{d x^2}(x+\epsilon z)\right]\, dx\right)\, dz\, .
\end{split}
\end{equation}
But, by the fundamental theorem of calculus, (\ref{eq:pressureIBP})  and Corollary \ref{cor:FirstDerivCorollary}
$$
\int_{x_1}^{x_2} \E\left[\frac{d^2 p_{\Lambda}}{d x^2}(x+\epsilon z)\right]\, dx\,
=\,
\E\left[\frac{dp_{\Lambda}}{d x}(x+\epsilon z)\right] \Bigg|_{x_1}^{x_2}\, ,
$$
so that
$$
\frac{\epsilon}{x_2-x_1}\, \int_{-1}^{1} (1 - |z|)
\left(\int_{x_1}^{x_2} \E\left[\frac{d^2 p_{\Lambda}}{d x^2}(x+\epsilon z)\right]\, dx\right)\, dz\,
\leq\,
\frac{2  \overline{c}_{\Lambda}' \epsilon}{x_2-x_1}\, \int_{-1}^{1} (1-|z|) (|x_1|+|x_2|+2\epsilon)\, dz\, .
$$
Using this with (\ref{eq:Penultimate}) gives the result
\begin{equation}
\frac{1}{x_2-x_1}\, \int_{x_1}^{x_2} \E\left[\left|\Omega_{\Lambda,x}\left(h'_{\Lambda}\right) -
\left\langle h'_{\Lambda} \right\rangle_{\Lambda,x}\right|\right]\, dx\,
\leq\,
2(|x_1|+|x_2|+2\epsilon)  \left(
\frac{\sqrt{2\overline{c}_\Lambda'}}{\epsilon \sqrt{\pi |\Lambda|}}+
\frac{\overline{c}_{\Lambda}' \epsilon}{x_2-x_1}\right)\, .
\end{equation}
Now we may optimize in $\epsilon > 0$. We choose $\epsilon = \sqrt{x_2-x_1} / (\overline{c}_{\Lambda}' |\Lambda|)^{1/4}$ to get
the desired result, equation (\ref{eq:DesideratumB}).
\end{proofof}

\section{Conclusions and perspectives}

A major new development in the subject of mean-field spin glasses is Panchenko's proof that strong stability
implies ultrametricity. More precisely, Panchenko in \cite{Panchenkoultra} proved that if a random probability measure $G$ in a Hilbert-space-valued inner-product satisfies
the Ghirlanda-Guerra identities in the \textsl{distributional sense}, then this implies ultrametricity. Let show how to combine the results of the  previous sections with the Panchenko's theorem.

A random measure $G$
satisfies the Ghirlanda-Guerra identities in the \textsl{distributional sense}  if for any $n\geq2$, any bounded measurable function
$f$ of the overlaps $(c_{l,l'})_{l,l' \leq n}$ and any bounded measurable function $\psi$  of one overlap, the following hold:
\be\label{eGGI}
\mathbb{E}\Omega\Big( f\psi(c_{1,2})\Big)=\frac{1}{n}\mathbb{E}\Omega\Big( f\Big)\mathbb{E}\Omega\Big(\psi(c_{1,2})\Big)+\frac{1}{n}\sum_{l=2}^n\mathbb{E}\Omega\Big( f\psi(c_{1,l})\Big)
\ee
where $\Omega$ is the average
with respect to $G^{\otimes\infty}$.
\newline
A random measure $G$ is said to be \textit{ultrametric} if the distribution of $(c_{l,l'})_{l,l' \geq 1}$ satisfies

\be\label{Ultra}
\mathbb{E}\Omega\Big( I\Big(c_{1,2}\geq\min(c_{1,3},c_{2,3})\Big)\Big)=1
\ee
By Theorem \ref{thm:main}, we can deduce that the asymptotic Gibbs measure of the model defined in (\ref{h_MP}) satisfies  the condition (\ref{eGGI}) for every power of the covariance and, by the Stone-Weierstrass theorem, this suffices to ensure that the Ghirlanda-Guerra identities hold in the \textsl{distributional sense}. More precisely Theorem \ref{thm:main} has two consequences. The proposition $(iv)$ implies that
\be\label{eGGI2}
\lim_{\Lambda \nearrow \Z^d}\mathbb{E}\Omega_{\Lambda,\beta}\Big( f\psi(c_{1,2})\Big)=\frac{1}{n}\mathbb{E}\Omega_{\Lambda,\beta}\Big( f\Big)\mathbb{E}\Omega_{\Lambda,\beta}\Big(\psi(c_{1,2})\Big)+\frac{1}{n}\sum_{l=2}^n\mathbb{E}\Omega_{\Lambda,\beta}\Big( f\psi(c_{1,l})\Big)
\ee
The proposition (ii) implies that this equation is true {\em for almost every} choice of parameter
$\beta = (\beta_p)_{p=1}^{\infty}$. In both case, we can identify the asymptotic Gibbs measure with a random  probability measure in a Hilbert-space using the representation theorem of Dovbysh and Sudakov (see \cite{ArguinChatterje} or \cite{Panchenkoover}) and the measure is  \textsl{ultrametric} by the Panchenko's result.\\

But it had been earlier known that the Ghirlanda-Guerra identities, when combined
with ultrametricity, imply that the precise distribution of the overlaps is given by the general class of Derrida-Ruelle
Poisson-Dirichlet random probability cascades. This result is know as  Baffioni-Rosati theorem (see
\cite{BaffioniRosati} and  \cite{Talagrand2} for a more complete exposition). One of its consequence
is that for an uncorrelated model (defined by having all the overlaps either $q_0<1$ or $1$)
the Ghirlanda-Guerra identities imply that it behaves as a single level Derrida-Ruelle cascade
known as the Poisson-Dirichlet random partition function. See for example \cite{Talagrand}.
Therefore, to summarize, Panchenko and Baffioni-Rosati results and Theorem \ref{thm:main}, implies that
a general Hamiltonian (\ref{h_alaruelle}) has an ultrametric Derrida Ruelle quenched equilibrium measure.\\

We have proved that all triangles of $c$-overlaps in a wide class of $d$-dimensional spin glass models must be isosceles or equilateral, as demanded by ultrametricity. This doesn't exclude the possibility of having only equilateral triangles, which is the trivial case for the overlap. We hope to return to study this question in the future. At present the issue on triviality is only studied numerically (see \cite{aum,um} for a recent debate).\\

It has been observed that for models such as the Edwards-Anderson model, the genuine overlap $q$ and the link overlap
$Q_{\text{link}}=c_{\Lambda}$, are equivalent in the sense that the conditional quenched distribution of one given the other is actually non-random \cite{oveq}. This has been numerically verified up to the computational limits in the citation and
contradicts the interpretation of previous work \cite{py,km}. The results presented in this paper suggest the possibility
to approach the replica equivalence conjecture by observing that not only are both $q$ and $Q_{\text{l}}$ ultrametric, but actually so is any choice of a $c$ matrix in the closed positive-cone algebra. This led to the idea of considering
$(1-t) q + t  Q_{\text{l}}$ and its properties for every $t \in [0,1]$. \\

We finally mention that the ideas outlined within this work can be viewed as classical
probability problem of the type called infinite divisibility. One may consider a general type of transformation: let
$$
\mathcal{T}_{\Lambda}^{x} \langle \cdots \rangle_{\Lambda}\,
=\, \langle \cdot \rangle_{\Lambda,x}\, .
$$
Then stochastic stability refers to stability under $\mathcal{T}_{\Lambda}^{x/|\Lambda|}$
in the limit $|\Lambda|\to\infty$: which we may refer to as $\mathcal{T}^x$
even though one must interpret this in terms of limits.
This is invariance of the expectation, since $\langle \cdots \rangle_{\Lambda}$
is defined to be $\E[\Omega_{\Lambda}(\cdot)]$.
But if one allows for a more general transformation,
which shifts by higher-order interactions, then one can also conclude
stochastic stability in distribution, not merely in expectation.

A fundamental concept in probability is
invariance under random transformations.
Stochastic stability should be viewed in that context.
Infinite divisibility refers to the property of being in the range of a transformation such as $\mathcal{T}^x$
for all $x$.
The quenched state for spin glasses satisfies invariance, which is stochastic stability.
But a more general question is to search for those distributions which satisfy infinite divisibility.
This could be a line of further investigation, elsewhere.

In our paper we have tried to demonstrate the following point, which we feel bears repeating.
Panchenko has showed that stochastic stability implies the ultrametric replica symmetry
breaking scenario \cite{Panchenkoultra} (which does also include the possibility of trivial ultrametricity,
where for any three replicas, the triangle formed by the intrinsic overlaps may be equilateral always).

Panchenko's proof uses a strong version of de Finetti's theorem.
But this applies because of permutation
invariance of the replicas, not the underlying model.
Indeed as we have shown, stochastic stability and its implications even apply to the Edwards-Anderson
model if it is perturbed by an arbitrarily small inclusion of higher order interactions.
Stochastic stability is a general tool, which in principle may be applicable to any disordered model in statistical mechanics including short-ranged spin glasses. We mention that using the language of metastates \cite{AW,NS} 
stochastic stability has been approached in \cite{ArguinDamron} where identities where proved
everywhere in the parameters using periodic boundary conditions and averaging over translations. It remains 
open to show that their analogue of the link overlap coincides with the usual one which is what we used here.
See also \cite{Barra} for an alternative derivation of the stochastic stability identities.

It is finally worth to stress that the addition of the higher orders interactions introduced in $(\ref{h_MP})$
is done in order to obtain the identities in distribution and not only in mean. The nature of those higher
orders terms appear to be formally similar to the p-spin interactions of the mean field case models.
In particular their interaction structure doesn't decay with distance but it's a sum normalised with the volume.
At the same time we notice that the core term is still a two point interaction made of nearest neighbouring
sites and as such it retains the topological information of the lattice dimension. While, as we properly explained,
we make no claim on the non-triviality issue of the models, we observe that the deformation method allows us
to reach not only the identities in distribution for the deformed model, but also to extend them to the original
model by sending the deformation parameter $\epsilon\to 0$ as specified by Corollary $(\ref{co43})$.\\

{\bf Acknowledgments} The authors thanks an anonymous referee for a question who led to the final remark
of the conclusion. SS was supported by an  NSA Young Investigator's grant.

\end{document}